\documentclass[journal]{IEEEtran}
\usepackage{graphicx,color}
\usepackage{cite}
\usepackage{setspace} 
\usepackage{amsmath}
\usepackage{amsmath,amsthm,amssymb}
\usepackage{multirow}
\usepackage{hhline}
\usepackage{epsfig}
\usepackage{subfigure}
\usepackage{epstopdf}
\usepackage{verbatim}
\usepackage{algorithm}
\usepackage{algorithmicx}
\usepackage{algpseudocode}
\usepackage{etoolbox}
\usepackage{cases}
\usepackage{csquotes}
\usepackage{mathtools,stmaryrd}
\SetSymbolFont{stmry}{bold}{U}{stmry}{m}{n}
\usepackage{bbm}
\usepackage{lettrine}

\newcommand{\suchthat}{\;\ifnum\currentgrouptype=16 \middle\fi|\;}

\makeatletter
\newcommand*{\indep}{%
  \mathbin{%
    \mathpalette{\@indep}{}%
  }%
}
\newcommand*{\nindep}{%
  \mathbin{
    \mathpalette{\@indep}{\not}
  }%
}
\newcommand*{\@indep}[2]{%
  \sbox0{$#1\perp\m@th$}
  \sbox2{$#1=$}
  \sbox4{$#1\vcenter{}$}
  \rlap{\copy0}
  \dimen@=\dimexpr\ht2-\ht4-.2pt\relax
  \kern\dimen@
  {#2}%
  \kern\dimen@
  \copy0 
} 
\makeatother

\makeatletter
\newcommand*{\algrule}[1][\algorithmicindent]{%
  \makebox[#1][l]{%
    \hspace*{.2em}
    \vrule height .75\baselineskip depth .25\baselineskip
  }
}

\newcount\ALG@printindent@tempcnta
\def\ALG@printindent{%
    \ifnum \theALG@nested>0
    \ifx\ALG@text\ALG@x@notext
    \else
    \unskip
    \ALG@printindent@tempcnta=1
    \loop
    \algrule[\csname ALG@ind@\the\ALG@printindent@tempcnta\endcsname]%
    \advance \ALG@printindent@tempcnta 1
    \ifnum \ALG@printindent@tempcnta<\numexpr\theALG@nested+1\relax
    \repeat
    \fi
    \fi
}
\patchcmd{\ALG@doentity}{\noindent\hskip\ALG@tlm}{\ALG@printindent}{}{\errmessage{failed to patch}}
\patchcmd{\ALG@doentity}{\item[]\nointerlineskip}{}{}{} 
\makeatother

\newtheorem{prop}{Proposition}
\newtheorem{corr}{Corollary}

\allowdisplaybreaks

\begin{document}

\title{Multi-dimensional Lorenz-Based Chaotic Waveforms for Wireless Power Transfer}

\author{Authors}
\author{Priyadarshi Mukherjee, \textit{Member, IEEE}, Constantinos Psomas, \textit{Senior Member, IEEE}, and~Ioannis Krikidis, \textit{Fellow, IEEE}
\thanks{P. Mukherjee, C. Psomas, and I. Krikidis are with the Department of Electrical and Computer Engineering, University of Cyprus, Nicosia 1678 (E-mail: \{mukherjee.priyadarshi, psomas, krikidis\}@ucy.ac.cy).

This work was co-funded by the European Regional Development Fund and the Republic of Cyprus through the Research and Innovation Foundation, under the projects INFRASTRUCTURES/1216/0017 (IRIDA) and POST-DOC/0916/0256 (IMPULSE). This work has also received funding from the European Research Council (ERC) under the European Union's Horizon 2020 research and innovation programme (Grant agreement No. 819819).}}

\maketitle
\begin{abstract}
In this paper, we investigate multi-dimensional chaotic signals with respect to wireless power transfer (WPT). Specifically, we analyze a multi-dimensional Lorenz-based chaotic signal under a WPT framework. By taking into account the nonlinearities of the energy harvesting process, closed-form analytical expressions for the average harvested energy are derived. Moreover, the practical limitations of the high power amplifier (HPA) at the transmitter are also taken into consideration. We interestingly observe that for these types of signals, high peak-to-average-power-ratio (PAPR) is not the only criterion for obtaining enhanced WPT. We demonstrate that while the HPA imperfections do not significantly affect the signal PAPR, it notably degrades the energy transfer performance. As the proposed framework is general, we also demonstrate its application with respect to a H\'enon signal based WPT. Finally we compare Lorenz and H\'enon signals with the conventional multisine waveforms in terms of WPT performance.
\end{abstract}


\IEEEpeerreviewmaketitle

\section{Introduction}

\lettrine[lines=2]{E}{merging} technologies such as the Internet of Things and machine type communications are expected to support a massive number of wireless devices. This becomes a critical issue given their limitations in terms of energy resources. In this context, wireless power transfer (WPT) is nowadays attracting more and more attention \cite{nmag}, which refers to the transmission of dedicated radio-frequency (RF) signals to wirelessly charge devices, located at a certain distance from the transmitter. WPT is essentially based on the efficient conversion of the received RF signals to direct current (DC). This RF-to-DC conversion is typically done by a rectifying antenna (rectenna) circuit at the receiver. The rectenna harvests electromagnetic energy from the RF signal, followed by rectification and filtering by means of a rectifier and a low pass filter (LPF), respectively \cite{wdesg}. For obtaining an enhanced WPT performance, the choice and design of the transmitted signal is equally important as the transmission power. Even though multitone waveforms result in better energy transfer performance with increasing number of tones \cite{wdesg}, the authors in \cite{bweff} claim that these signals can improve the WPT performance only if  the frequency-spacing between the neighbouring tones is properly selected. In this context, several other signal designs have also been investigated.

The authors in \cite{papr}, demonstrate that the nonlinearity at the rectenna causes waveforms with high peak-to-average-power-ratio (PAPR) to provide a higher DC output.  The works in \cite{chaosexp1} and \cite{chaosexp2}, propose the use of chaotic signals to improve the RF-to-DC conversion efficiency, where it is experimentally observed that chaotic signals are beneficial for WPT. The authors in \cite{jstsp}, propose a differential chaos shift keying (DCSK)-based WPT architecture and an associated DCSK-based waveform, solely for WPT. Moreover, these works also illustrate that channel fading further enhances WPT capability of chaotic signal-based waveforms. The above studies focus on the WPT performance of DCSK, which is an one dimensional chaotic signal. Even though experimental results demonstrate that multi-dimensional chaotic signals are beneficial for WPT \cite{chaosexp1,chaosexp2}, no analytical framework exists, to support this claim. The specific gains in energy harvesting (EH) from multi-dimensional chaotic signals have not been properly explored.

Motivated by this, in this letter, we investigate multi-dimensional chaotic signal-based waveform designs for WPT. In particular, we study a point-to-point WPT system, where a Lorenz signal generator is used at the transmitter. We characterize the harvested energy in terms of the parameters of the transmitted waveform. Our study also takes into account an EH model based on the nonlinearities of the rectification process. Even though the analytical results provided are specifically for the Lorenz system, the framework is general. For the sake of completeness, we have also demonstrated briefly how the system model can be extended to a H\'enon system, which is another class of multi-dimensional chaotic systems. Extensive Monte Carlo simulations validate the analysis. Finally, we compare the WPT performance of the Lorenz and H\'enon signal with the conventional $N$-tone multisine waveforms. To the best of our knowledge, this is the first work that presents a complete analytical framework for multi-dimensional chaotic signal-based WPT.

\section{System Model}
Consider a point-to-point WPT set-up, where the transmitter is a Lorenz signal generator and the receiver consists of an EH circuit that employs a rectifier to convert the received signal to DC. A conventional Lorenz system at the transmitter is defined as a three-dimensional nonlinear dynamical system  \cite{lorenz}
\vspace{-1mm}
\begin{subequations} \label{ldef}
\begin{align} 
\dot{x}(t)&=\sigma(y(t)-x(t)), \label{ldef1}\\
\dot{y}(t)&=x(t)(r-z(t))-y(t), \label{ldef2}\\
\dot{z}(t)&=x(t)y(t)-\beta z(t), \label{ldef3}
\end{align}
\end{subequations}
where $x(t),y(t),z(t)$ are state variables, $\sigma,r,\beta>0$ are control parameters, and $\dot{x}(t),\dot{y}(t),$ and $\dot{z}(t)$ denote $\frac{dx(t)}{dt},\frac{dy(t)}{dt},$ and $\frac{dz(t)}{dt}$, respectively. When a Lorenz-based transmitter is designed, the values of various circuit components (i.e. the resistors and capacitors) at the transmitter are controlled by
$\sigma,r,$ and $\beta$, respectively.

At the receiver side, the received noisy signal is
\vspace{-2mm}
\begin{equation}  \label{rcvd}
\tilde{x}(t)=\sqrt{P_td^{-\alpha}}x(t)+n,
\end{equation}
where $P_t$ is the transmission power, $d$ is the transmitter-receiver distance, $\alpha$ is the pathloss exponent, and $n$ is the additive white Gaussian noise (AWGN)\footnote{In this work, we focus on the aspect of waveform design. The consideration of fading is straightforward as we would only need its moments \cite{jstsp}.}. It is worth noting that one of the main aspects of the Lorenz signal, when used for the purpose of information transfer, is the synchronization issue at the receiver \cite{dyran}. However, when used for purely energy transfer, we do not need any synchronization block at the receiver.

By considering a circuit-based nonlinear model of the harvester circuit, the output DC current is approximated in terms of the received signal $\tilde{x}(t)$ as \cite{wdesg}
\vspace{-2mm}
\begin{equation} \label{brunoeh}
\eta_{\rm MC}=k_2R_{ant}\mathbb{E} \{ |\tilde{x}(t)|^2 \}+k_4R_{ant}^2\mathbb{E} \{ |\tilde{x}(t)|^4 \},
\end{equation}
where the subscript MC refers to the multi-dimensional chaotic transmission. $\eta_{\rm MC}$ is a monotonically increasing function of the DC component of the current at the harvester output and the parameters $k_2,k_4,$ and $R_{ant}$ are determined by the circuit characteristics. Any RF energy harvesting from the AWGN is considered to be negligible and thus it is ignored \cite{wdesg}.

\section{Lorenz signal-based WPT}

We investigate the effect of the Lorenz circuit-based chaotic signal on WPT. The general approach to gain insights on the WPT performance of any multi-dimensional chaotic system, is to obtain its equilibrium points and investigate the system behavior around them \cite{eqlb}. As we deal with the Lorenz system, we first obtain its corresponding equilibrium points\footnote{Although the proposed framework is specifically for the Lorenz system, it can be extended to any multi-dimensional chaotic circuit, namely Colpitts oscillator, R{\"o}ssler system, or H\'enon system \cite{henon} (see Section III-C).}. Accordingly, we evaluate the harvested DC $\eta_{\rm MC}$ in terms of the received signal $\tilde{x}(t)$ and its associated control parameters, i.e. $\sigma,r,$ and $\beta$.

It should be noted that the practical implementation of \eqref{ldef} by means of an electronic circuit is a complicated process due to the wide dynamic range of the state variables $x(t),y(t),$ and $z(t)$ \cite{dyran}. This bottleneck arises due to the limited transmission capability of a practical high power amplifier (HPA) at the transmitter, which is resolved by introducing a technique of scaling, i.e. we scale the state variables  \eqref{ldef} as
\vspace{-2mm}
\begin{equation}    \label{scaling}
x_{\rm sc}(t)=\frac{x(t)}{\epsilon_x}, \quad y_{\rm sc}(t)=\frac{y(t)}{\epsilon_y}, \quad \text{and} \qquad z_{\rm sc}(t)=\frac{z(t)}{\epsilon_z},
\end{equation}
where $\epsilon_x,\epsilon_y,\epsilon_z$ $\in [1,\infty)$ are scaling parameters.

\subsection{Steady state analysis}    \label{ssa}

By replacing \eqref{scaling} in \eqref{ldef}, we obtain the scaled Lorenz system
\vspace{-2mm}
\begin{subequations} \label{sldef}
\begin{align} 
\dot{x}_{\rm sc}(t)&=\sigma\left( \frac{\epsilon_y}{\epsilon_x}y_{\rm sc}(t)-x_{\rm sc}(t) \right), \label{sldef1}\\
\dot{y}_{\rm sc}(t)&=\frac{\epsilon_x}{\epsilon_y} x_{\rm sc}(t) \left( r-\epsilon_zz_{\rm sc}(t) \right) -y_{\rm sc}(t), \label{sldef2}\\
\dot{z}_{\rm sc}(t)&=\frac{\epsilon_x\epsilon_y}{\epsilon_z} x_{\rm sc}(t)y_{\rm sc}(t)-\beta z_{\rm sc}(t). \label{sldef3}
\end{align}
\end{subequations}
To gain insights on the properties of this modified Lorenz system, we need to locate the equilibrium points, which is computed by setting $\dot{x}_{\rm sc}(t)=\dot{y}_{\rm sc}(t)=\dot{z}_{\rm sc}(t)=0$\cite{eqlb}.
\begin{prop}    \label{prop1}
The equilibrium points $p_{\rm sc}=\left(x_{\rm eq},y_{\rm eq},z_{\rm eq} \right)$ for a scaled Lorenz system are
\vspace{-1mm}
\begin{subnumcases}{\label{seqp} p_{\rm sc}=}
(0,0,0), \label{seqp1}\\
\left( \frac{\sqrt{\beta(r-1)}}{\epsilon_x},\frac{\sqrt{\beta(r-1)}}{\epsilon_y},\frac{r-1}{\epsilon_z}\right) ,  \label{seqp2} \\
\left( -\frac{\sqrt{\beta(r-1)}}{\epsilon_x},-\frac{\sqrt{\beta(r-1)}}{\epsilon_y},\frac{r-1}{\epsilon_z}\right) . \label{seqp3}
\end{subnumcases}
\end{prop}
\begin{proof}
See Appendix \ref{app1}.
\end{proof}

To further obtain insights into an $n$-dimensional nonlinear system, we fix any $(n-1)$ parameters and observe the system behaviour with respect to the remaining parameter \cite{sparrow}. Accordingly, in this case, we fix $\sigma,\beta$ and vary $r$.

We observe from \eqref{seqp2} and \eqref{seqp3}, that these equilibrium points do not exist for $r<1$, as they must be real \cite{eqlb}. Furthermore, to have a stable system, $r$ can not attain any arbitrary value greater than unity. Hence, we state the following proposition that refers to the range of $r$ that ensures system stability.

\begin{prop}    \label{prop2}
Equilibrium points \eqref{seqp2} and \eqref{seqp3} are stable, if and only if,
\vspace{-1mm}
\begin{equation}
r \in \left( 1,\frac{\sigma(\sigma+\beta+3)}{\sigma-\beta-1}\right).
\end{equation}
\end{prop}
\begin{proof}
See Appendix \ref{app2}.
\end{proof}

It is worth noting that $\epsilon_x,\epsilon_y,\epsilon_z$ do not have any effect on the domain of $r$ as stated in Proposition \ref{prop2}. In general, Lorenz chaotic systems are very sensitive to the initial conditions $p_{\rm in}=(x_0,y_0,z_0)$. This implies that two trajectories starting very close together rapidly diverge after a very short span of time \cite{sparrow}. However, Proposition \ref{prop1} and \ref{prop2} state that when $r \in \left( 1,\frac{\sigma(\sigma+\beta+3)}{\sigma-\beta-1}\right)$, the system asymptotically attains \eqref{seqp2} or \eqref{seqp3} in finite time, irrespective of $p_{\rm in}$.

As illustrated in Fig. \ref{fig:syseff}, we can not accurately say to which equilibrium point, i.e. \eqref{seqp2} or \eqref{seqp3}, will $p_{\rm in}$ reach at time $t \rightarrow \infty$; however, we can state that $p_{\rm in}$ reaches one of them. The figure demonstrates an unscaled Lorenz system, i.e. $\epsilon_x=\epsilon_y=\epsilon_z=1$, for the sake of presentation.
\begin{figure}[!t]
  \centering
  \subfigure[Effect of $r$.]
  {{\includegraphics[height=2.8cm,width=0.492\linewidth]{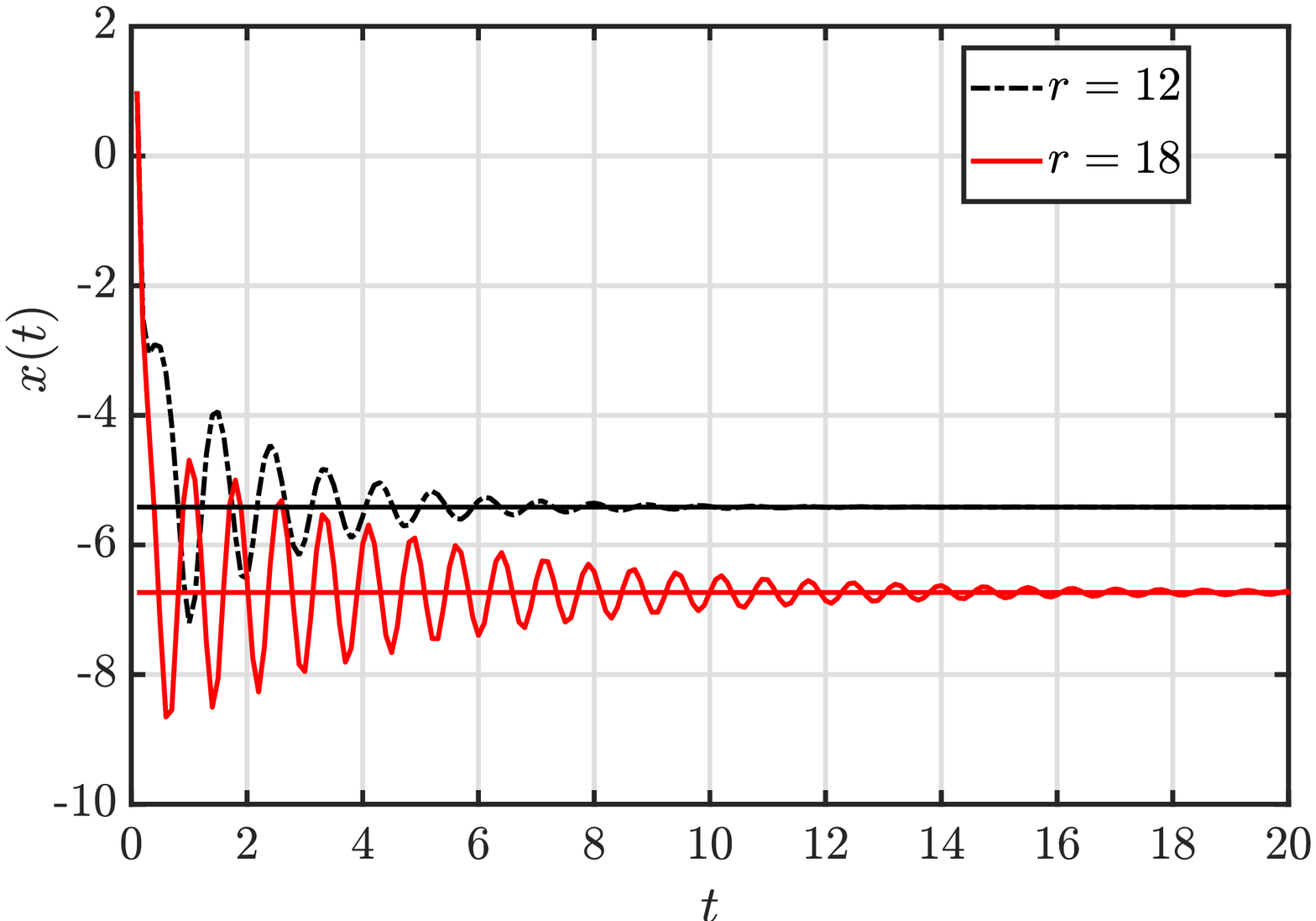} }
    \label{fig:reff}}
    \hspace{-6mm}
  \subfigure[Effect of initial point.]   
  {{\includegraphics[height=2.8cm,width=0.492\linewidth]{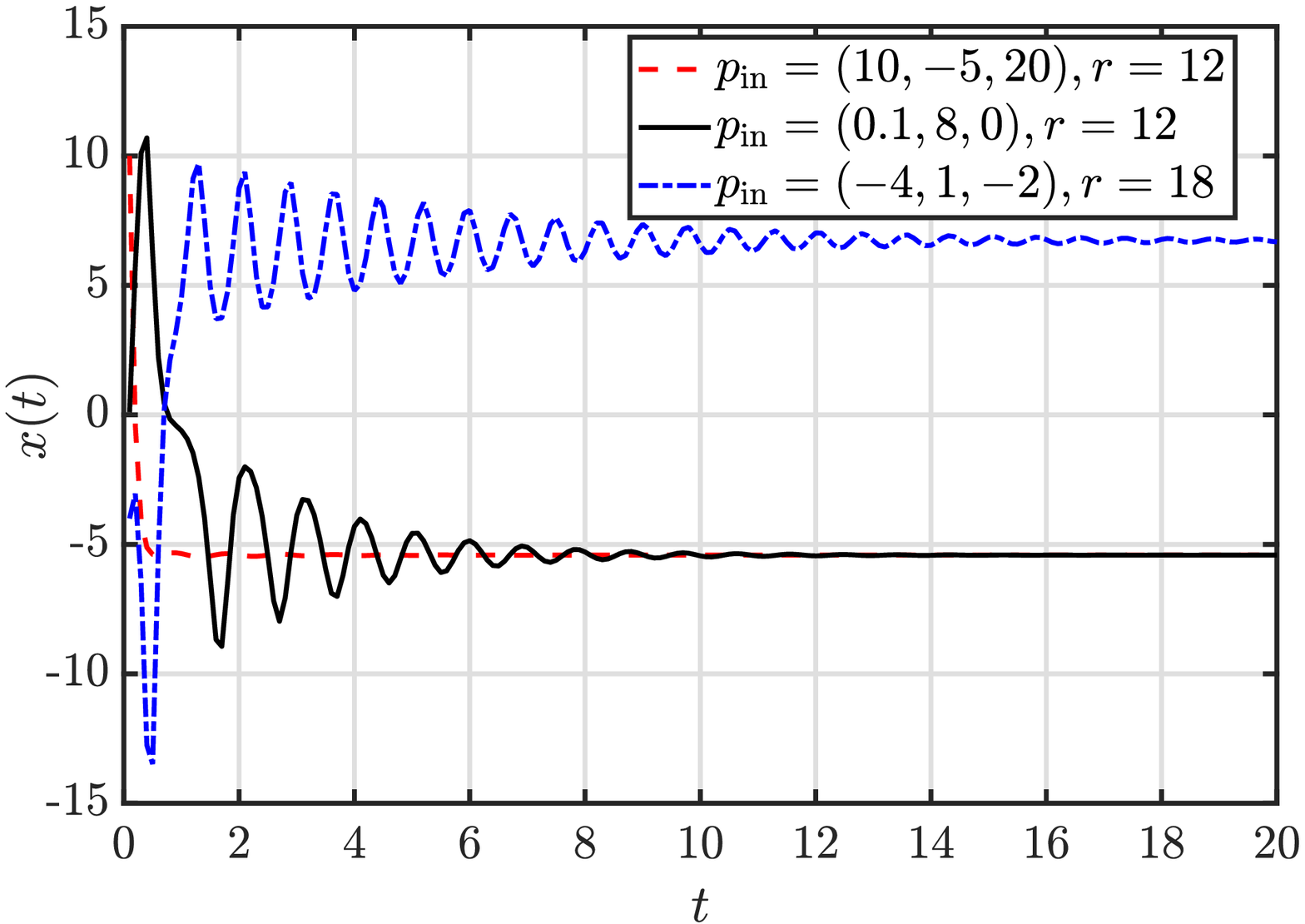} }
   \label{fig:ineff}  }
   \vspace{-2.4mm}
    \caption{Effects of system parameters; $\sigma\!=\!10,\beta\!=\!8/3,$ and $\epsilon_x\!=\!\epsilon_y\!=\!\epsilon_z\!=\!1$.}
    \label{fig:syseff}
\vspace{-4mm}
\end{figure}
Fig. \ref{fig:reff} depicts the effect of varying $r$ on $x(t)$, for $p_{\rm in}=(1,-5,20)$. We observe that even for identical $p_{\rm in}$, $x(t)$ converges to different values for different $r$; the convergence rate follows a decreasing trend with increasing $r$. Specifically, while $x(t)$ corresponding to $r=12$ converges at $t<10$, $x(t)$ corresponding to $r=18$ converges at $t\approx20$ for the same $p_{\rm in}$. Furthermore, it is interesting to observe from Fig. \ref{fig:ineff} that different $p_{\rm in}$ results in identical steady state values at $t \rightarrow \infty$ for identical $r$. Note that both $p_{\rm in}=(10,-5,20)$ and $p_{\rm in}=(0.1,8,0)$, with $r=12$, attain identical $\lim\limits_{t \rightarrow \infty} x(t)$. This corroborates the previous claim, i.e. when the system is stable, we have $\lim\limits_{t \rightarrow \infty} x(t) \in \left\lbrace -\frac{\sqrt{\beta(r-1)}}{\epsilon_x},\frac{\sqrt{\beta(r-1)}}{\epsilon_x} \right\rbrace $ irrespective of $p_{\rm in}$, which is a constant quantity. However, this uncertainty does not affect the characterization of the harvested DC, as discussed in the next section.

\subsection{Harvested energy}

From \eqref{rcvd} and \eqref{brunoeh}, we observe that $\eta_{\rm MC}$ is a function of the $x$-component of the equilibrium points \eqref{seqp2} and \eqref{seqp3}. Hence, we state the following proposition.
\begin{prop}  \label{prop3}
The steady state harvested energy for a scaled Lorenz system is
\vspace{-2mm}
\begin{equation}  \label{sharvchaos}
\eta_{\rm SL}=\frac{\rho_1\beta(r-1)}{\epsilon_x^2}+\frac{\rho_2\beta^2(r-1)^2}{\epsilon_x^4},
\end{equation}
where $\rho_1=d^{-\alpha}k_2R_{ant}P_t$ and $\rho_2=d^{-2\alpha}k_4R_{ant}^2P_t^2$.
\end{prop}
\begin{proof}
See Appendix \ref{app3}.
\end{proof}
We have not considered any specific bandwidth for this work. However the authors in \cite{nlor} demonstrate that the Lorenz signal is essentially  narrowband in nature. Hence, the effect of fading can be taken into account by multiplying the first and second terms of $\eta_{\rm SL}$ in \eqref{sharvchaos} with the second and fourth (raw) moments of the channel fading gain \cite{jstsp}, respectively.
Proposition \ref{prop3} provides a generalized closed-form expression for the harvested energy in terms of the scaling parameters\footnote{The energy transfer model considered in this work is based on the assumption that the harvester operates in the nonlinear region \cite{wdesg}. If the power of the harvester input signal becomes too large, the diode inside the harvester is forced into the saturation region of operation, making the derived analytical results inapplicable.}. The harvested DC  $\eta_{\rm SL}$ corresponding to an ideal scenario, i.e. $\epsilon_x\!=\!\epsilon_y\!=\!\epsilon_z \rightarrow 1$, can also be obtained as a special case, given in the following corollary.
\begin{corr}
In an ideal scenario, $\eta_{\rm SL}$ is given by
\vspace{-2mm}
\begin{equation}  \label{hel}
\eta_{\rm L}=\rho_1\beta(r-1)+\rho_2\beta^2(r-1)^2.
\end{equation}
\end{corr}
\noindent The above corollary follows directly from Proposition \ref{prop3}, by considering $\epsilon_x \rightarrow 1$. Techniques such as predistortion \cite{wcl} enable to overcome the HPA imperfections, which further enhances the WPT performance.

\begin{figure}[!t]
\centering\includegraphics[width=0.76\linewidth]{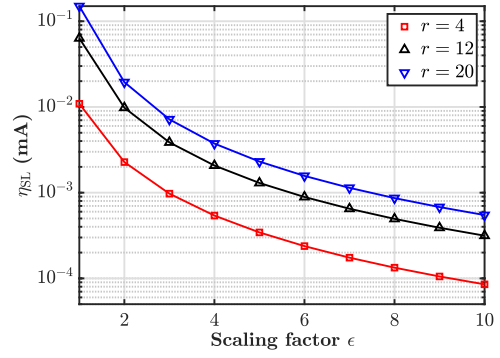}
   \vspace{-2.4mm}
\caption{Effect of scaling parameter on WPT performance; lines correspond to analysis and markers correspond to simulation results.}
\label{fig:flatharv}
\vspace{-4mm}
\end{figure}

Note that, we obtain closed-form  expressions of the harvested DC for the stable operation region of the scaled Lorenz attractor circuit. However, for $r \notin \left( 1,\frac{\sigma(\sigma+\beta+3)}{\sigma-\beta-1}\right)$, i.e. in case of an unstable Lorenz system, \eqref{sldef} does not converge to any specific point. As a result, it is not possible to analytically characterize the harvested energy for this region of operation.

Even though the above analytical results are based on the Lorenz system, the provided framework is general. Irrespective of the system, what is required is the set of equilibrium points and the corresponding stability region. To show this, we provide in the next section, how the approach can be extended to the H\'enon system \cite{henon}.

\subsection{Extension to other multi-dimensional chaotic waveforms}

The approach already presented to analytically characterize the harvested energy is general and can be applied to any chaotic system. For the sake of completeness, we consider another chaotic system, i.e. the H\'enon map. A conventional H\'enon map is defined by the following first-order difference equations \cite{henon}
\vspace{-2mm}
\begin{subequations} \label{hdef}
\begin{align} 
x_{n+1}&=y_n+1-\gamma x_n^2, \label{hdef1}\\
y_{n+1}&=\delta x_n, \label{hdef2}
\end{align}
\end{subequations}
where $x,y$ are state variables and $\gamma,\delta$ are control parameters. Note that, unlike the Lorenz system, a typical H\'enon map is governed by two state variables. Following an identical approach as before, we obtain the corresponding stable equilibrium points as
\begin{align}
p_{\rm henon}&=\left( \left(\frac{\delta-1 + \sqrt{(1-\delta)^2+4\gamma}}{2\gamma} \right) \right., \nonumber \\
& \qquad\left. \delta\left(\frac{\delta-1 + \sqrt{(1-\delta)^2+4\gamma}}{2\gamma}\right) \right),
\end{align}
with $0<\delta<1,$ $\gamma \neq 0,$ and $\displaystyle \gamma \in \left( -\frac{1}{4}(1-\delta)^2, \frac{3}{4}(1-\delta)^2 \right)$. As a result, the harvested DC in this case is obtained as
\vspace{-1mm}
\begin{equation}  \label{heh}
\eta_{\rm H}=\rho_1 \Phi^2 + \rho_2 \Phi^4,
\end{equation}
where $\Phi=\frac{\delta-1 + \sqrt{(1-\delta)^2+4\gamma}}{2\gamma}$. Moreover, if we consider a scaled H\'enon system to accommodate the HPA imperfections, a transformation of variables as described in \eqref{scaling} is required, and the equilibrium point and the stability region changes accordingly. As was in the case of a Lorenz system, it is not possible to analytically characterize the harvested energy for an unstable H\'enon system. Finally, we observe from \eqref{hel} and \eqref{heh}, it is possible to obtain $\eta_{\rm L}>\eta_{\rm H}$, only when
\vspace{-1mm}
\begin{equation}
\sqrt{\beta(r-1)}>\frac{\delta-1 + \sqrt{(1-\delta)^2+4\gamma}}{2\gamma}.
\end{equation}
\section{Numerical Results}

We consider a transmission power $P_t=30$ dBm, transmitter-receiver distance $d=20$ m, pathloss exponent $\alpha=4$, and system parameters $\sigma=10$ and $\beta=8/3$. This implies that the system is stable as long as $1<r<\frac{\sigma(\sigma+\beta+3)}{\sigma-\beta-1}=24.74$. Note that we can also consider any other value of $\sigma$ and $\beta$, which accordingly gets reflected in the range of $r$, as stated in Proposition \ref{prop2}, to have a stable system. Finally, the parameters considered for the WPT model are $k_2\!=\!0.0034,k_4\!=\!0.3829,$ and $R_{ant}\!=\!50$ $\Omega$, respectively \cite{wdesg}. In addition, all the presented results are accompanied by Monte Carlo simulations by using $10^6$ realizations.

Fig. \ref{fig:flatharv} demonstrates the effect of $r$ and scaling parameters $\epsilon_x=\epsilon_y=\epsilon_z=\epsilon$ on the WPT performance with initial point $p_{\rm in}=(4,-10,0.6)$. We observe that the theoretical results (lines) match very closely with the simulation results (markers); this verifies our proposed analytical framework. The figure illustrates that the harvested DC increases with $r$, which can not be indefinitely increased due to the stability aspect. Furthermore, the figure also demonstrates the effect of $\epsilon$ on the EH performance of the receiver. We observe that $\epsilon$ strongly affects the EH performance, when all the other system parameters remain constant. Recall that the scaling parameter depends on the practical imperfections of the HPA at the transmitter. This implies that if we are able to overcome the HPA imperfections, it enhances the harvested DC under identical scenarios.

\begin{figure}[!t]
  \centering
  \subfigure[Effect on PAPR.]
  {{\includegraphics[width=0.492\linewidth]{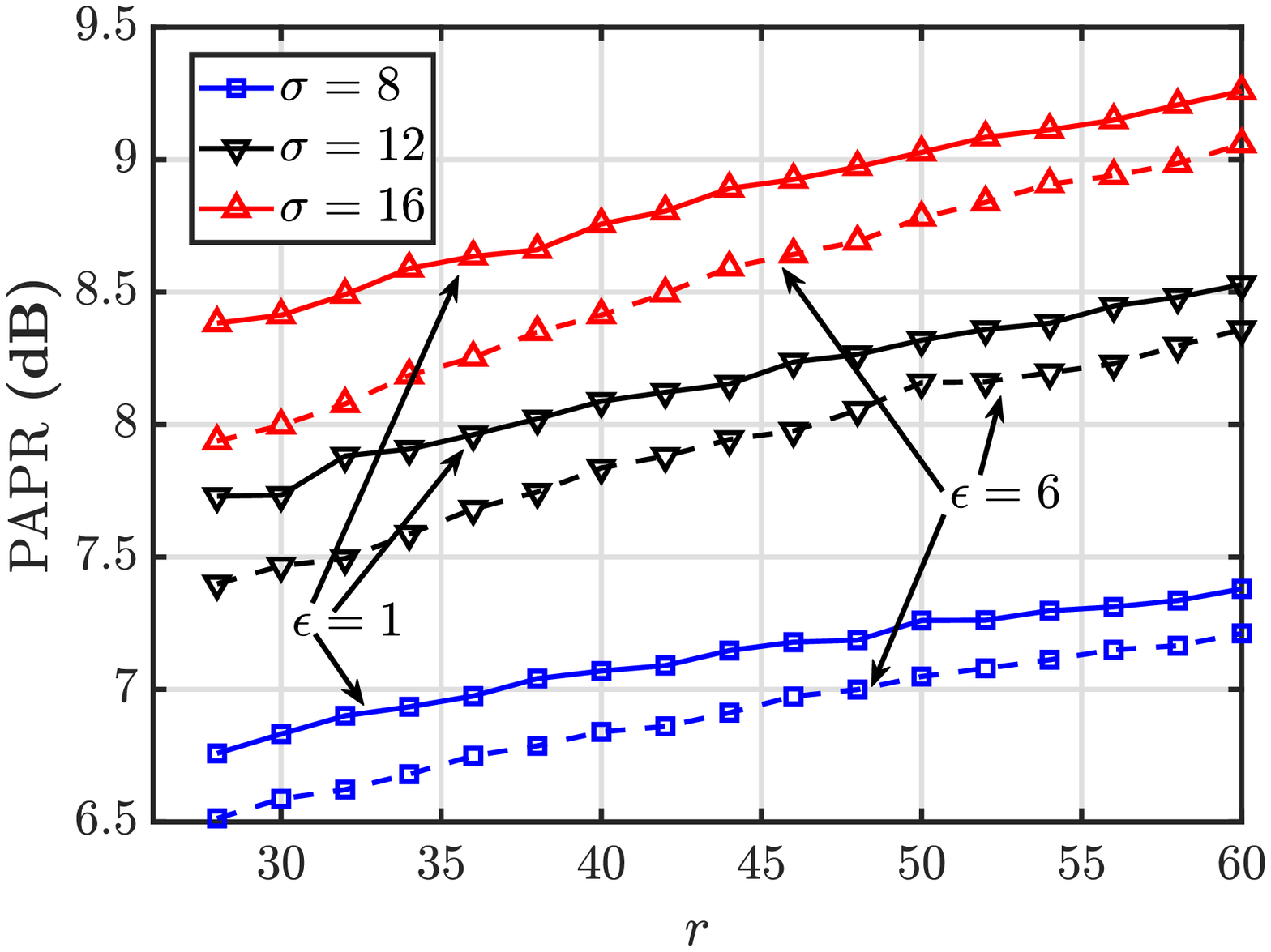} }
    \label{fig:unstab1}}
    \hspace{-6mm}
  \subfigure[Effect on harvested energy.]   
  {{\includegraphics[width=0.492\linewidth]{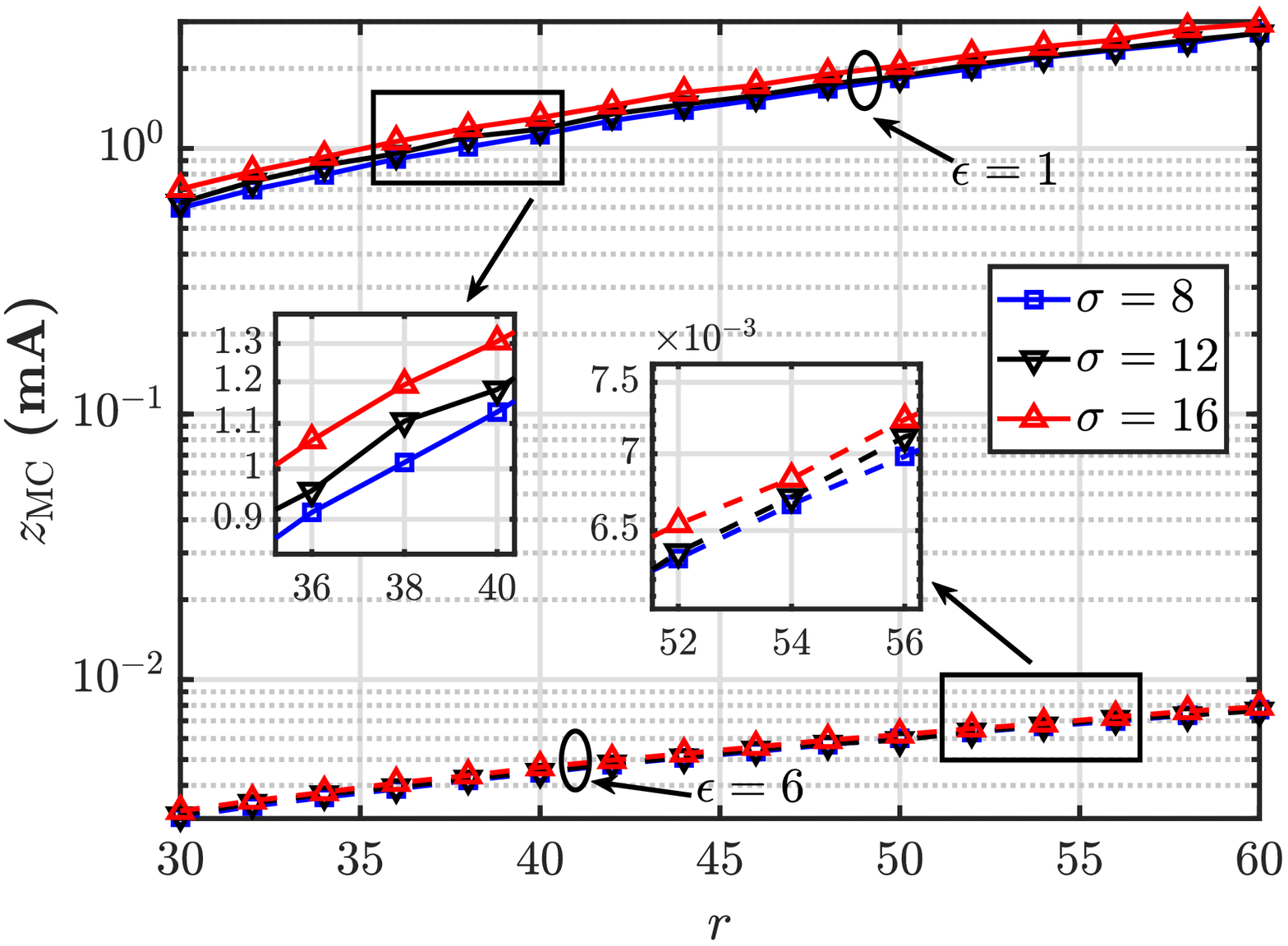} }
   \label{fig:unstab2}  }
      \vspace{-2.4mm}
    \caption{WPT performance for unstable system with $p_{\rm in}=(0.1,10,0.1)$.}
    \label{fig:unstab}
\vspace{-4mm}
\end{figure}

Note that the analytical expressions derived for the harvested DC correspond to the stable region of operation for the Lorenz system. However, when the Lorenz system operates in the unstable region, i.e. $r>\frac{\sigma(\sigma+\beta+3)}{\sigma-\beta-1}=24.74$, the system does not converge to any particular value and as a result, analytical closed-form expressions can not be derived. We observe from Fig. \ref{fig:unstab1} that the PAPR of the signal at the receiver increases with $r$. Moreover, similar to \cite{papr}, Fig. \ref{fig:unstab2} illustrates that the harvested DC increases with PAPR. It is worthy to note that while the scaling ($\epsilon_x\!=\!\epsilon_y\!=\!\epsilon_z\!=\!6$ in the figure) does not significantly affect the Lorenz system in terms of PAPR of the signal at the receiver, there is a considerable gap in the WPT performance of the scaled and unscaled Lorenz system. This detrimental effect of scaling on the WPT performance justifies Proposition $3$ and the insights obtained thereafter. Moreover, the figure also illustrates that unlike to the stable scenario, $\sigma$ has an impact on both the PAPR and WPT performance of the unstable Lorenz system. This particular observation leads to an interesting conclusion, i.e. for multi-dimensional waveforms, high PAPR is not equivalent to higher harvested energy and that the one-to-one PAPR-harvested energy mapping holds only when we are dealing with uni-dimensional signals.

By considering an ideal transmission scenario, i.e. unity scaling factor, Fig. 4 compares the WPT performance of the Lorenz and H\'enon signal in their respective stable region of operation, with the existing $N$-tone multisine waveforms [2]. We observe that the harvested energy increases with increasing $r$, as stated in Proposition 3, and also with increasing number of multitones. Moreover, we also observe from the figure that any arbitrary multi-dimensional chaotic waveform transmission does not guarantee enhanced WPT performance; for example, the performance gap for the H\'enon map at $P_t=20$ dBm between $\gamma=0.2,\delta=0.1$ and $\gamma=0.001,\delta=0.9$. Note that the values considered are solely for illustration. Our results demonstrate that the selection of parameter values is equally important as the choice of the chaotic waveform.

\section{Conclusion}
In this paper, we focused on the effects of multi-dimensional chaotic signals on WPT. We evaluated the performance of energy transfer for the Lorenz waveform under both, stable and unstable regions of operation, in terms of its parameters. Closed-form expressions for the average harvested energy were analytically derived and it was observed from exhaustive simulations, that for multi-dimensional chaotic signals, high PAPR is not the single criterion for obtaining an enhanced energy transfer performance. As the proposed framework is general in nature, we also illustrate its application with respect to a H\'enon chaotic system. Finally, we compare the Lorenz and the H\'enon signal with the multisine waveforms, where we observe that the choice of parameters for the chaotic waveforms is crucial in terms of WPT. An immediate extension of this work is to investigate the class of multi-dimensional chaotic signals in a generalized frequency selective scenario.

\begin{figure}[!t]
\centering\includegraphics[width=0.76\linewidth]{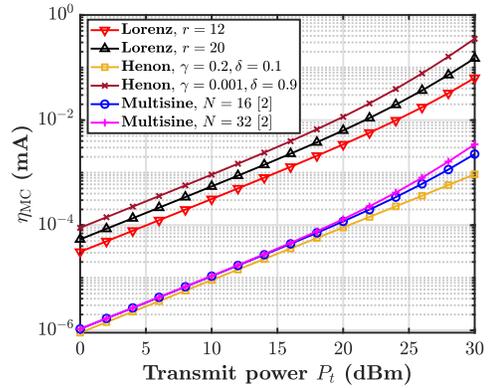}
   \vspace{-2.4mm}
\caption{Performance comparison.}
\label{fig:comp}
\vspace{-4mm}
\end{figure}

\appendices

\section{Proof of Proposition \ref{prop1}}
\label{app1}
By setting $\dot{x}_{\rm sc}(t)=\dot{y}_{\rm sc}(t)=0$\cite{eqlb} in \eqref{sldef}, we obtain
\vspace{-2mm}
\begin{equation}  \label{shrt}
y_{\rm sc}(t)=\frac{\epsilon_xx_{\rm sc}}{\epsilon_y} \quad \text{and} \quad z_{\rm sc}(t)=\frac{r-1}{\epsilon_z}.
\end{equation}
Similarly, by substituting $\dot{z}_{\rm sc}(t)=0$ in \eqref{sldef3} and combining with \eqref{shrt}, we have
\vspace{-2mm}
\begin{equation}  \label{sclxyz}
x_{\rm sc}(t)=\pm \frac{\sqrt{\!\beta(r\!-\!1)}}{\epsilon_x}, \: y_{\rm sc}(t)=\pm \frac{\sqrt{\!\beta(r\!-\!1)}}{\epsilon_y}, \: z_{\rm sc}(t)=\frac{r\!-\!1}{\epsilon_z}.
\end{equation}
Note that while \eqref{sclxyz} holds for $r>1$, the origin, i.e. $(0,0,0)$ is an equilibrium point for all values of $\sigma,r,\beta$. Thus, based on \eqref{shrt} and \eqref{sclxyz}, we obtain the equilibrium points for all possible values of the control parameters as
\vspace{-2mm}
\begin{subnumcases}{ p_{\rm sc}=}
(0,0,0), \\
\left( \frac{\sqrt{\beta(r-1)}}{\epsilon_x},\frac{\sqrt{\beta(r-1)}}{\epsilon_y},\frac{r-1}{\epsilon_z}\right) , \\
\left( -\frac{\sqrt{\beta(r-1)}}{\epsilon_x},-\frac{\sqrt{\beta(r-1)}}{\epsilon_y},\frac{r-1}{\epsilon_z}\right) .
\end{subnumcases}

\section{Proof of Proposition \ref{prop2}}
\label{app2}
We obtain the Jacobian matrix corresponding to \eqref{sldef} as \cite{papoulis}
\vspace{-2mm}
\begin{equation} \label{jsc}
J_{\rm sc}=\begin{bmatrix}
-\sigma & \displaystyle \frac{\epsilon_y}{\epsilon_x}\sigma & 0\\
\displaystyle \frac{\epsilon_x}{\epsilon_y}(r-\epsilon_zz_{\rm sc}(t)) & -1 & -\displaystyle \frac{\epsilon_x\epsilon_z}{\epsilon_y}x_{\rm sc}(t) \\
\displaystyle \frac{\epsilon_x\epsilon_y}{\epsilon_z}y_{\rm sc}(t) & \displaystyle \frac{\epsilon_x\epsilon_y}{\epsilon_z}x_{\rm sc}(t) & -\beta
\end{bmatrix}.
\end{equation}
By replacing the equilibrium points given by
\vspace{-2mm}
\begin{equation}
P_1=\left( \frac{\sqrt{\beta(r-1)}}{\epsilon_x},\frac{\sqrt{\beta(r-1)}}{\epsilon_y},\frac{r-1}{\epsilon_z}\right),
\end{equation}
\vspace{-2mm}
\begin{equation}
\text{and} \quad P_2=\left( -\frac{\sqrt{\beta(r-1)}}{\epsilon_x},-\frac{\sqrt{\beta(r-1)}}{\epsilon_y},\frac{r-1}{\epsilon_z}\right),
\end{equation}
in \eqref{jsc}, we obtain
\vspace{-2mm}
\begin{equation} \label{jeq1}
J_{\rm sc}^{(1)}\!\!=\!\!
\begin{bmatrix}
-\sigma & \displaystyle\frac{\epsilon_y}{\epsilon_x}\sigma & 0\\
\displaystyle \frac{\epsilon_x}{\epsilon_y} & -1 & -\displaystyle \frac{\epsilon_z}{\epsilon_y}\sqrt{\beta(r-1)} \\
\displaystyle \frac{\epsilon_x}{\epsilon_z}\sqrt{\beta(r-1)} & \displaystyle \frac{\epsilon_y}{\epsilon_z}\sqrt{\beta(r-1)} & -\beta
\end{bmatrix},
\end{equation}
and
\vspace{-2mm}
\begin{equation} \label{jeq2}
J_{\rm sc}^{(2)}\!\!=\!\!
\begin{bmatrix}
-\sigma & \displaystyle\frac{\epsilon_y}{\epsilon_x}\sigma & 0\\
\displaystyle \frac{\epsilon_x}{\epsilon_y} & -1 & \displaystyle \frac{\epsilon_z}{\epsilon_y}\sqrt{\beta(r-1)} \\
\displaystyle \frac{-\epsilon_x}{\epsilon_z}\sqrt{\beta(r-1)} & \displaystyle \frac{-\epsilon_y}{\epsilon_z}\sqrt{\beta(r-1)} & -\beta
\end{bmatrix},
\end{equation}
respectively. It is interesting to observe that the eigenvalues of both \eqref{jeq1} and \eqref{jeq2} are the roots of
\vspace{-2mm}
\begin{equation}
\lambda^3+\lambda^2(\sigma+\beta+1)+\lambda\beta(\sigma+r)+2\sigma\beta(r-1)=0.
\end{equation}
According to the Hurwitz stability criterion \cite{stab}, the system is stable, if and only if, all the principal diagonal minors of
\vspace{-2mm}
\begin{equation}
H=\begin{bmatrix}
\sigma+\beta+1 & 2\sigma\beta(r-1) & 0\\
1 & \beta(\sigma+r) & 0 \\
0 & \sigma+\beta+1 & 2\sigma\beta(r-1)
\end{bmatrix},
\end{equation}
are positive, which gives
\vspace{-2mm}
\begin{equation}
(\sigma+\beta+1)\beta(\sigma+r)>2\sigma\beta(r-1) \implies r<\frac{\sigma(\sigma+\beta+3)}{\sigma-\beta-1}.
\end{equation}
Moreover, we also observe from \eqref{ldef} that $\sigma,\beta,r$ are positive quantities, i.e. the inequality $\sigma>\beta+1$ must also be satisfied. This completes the proof.

\section{Proof of Proposition \ref{prop3}}
\label{app3}

As stated in Section \ref{ssa}, in the steady state, the system attains $P_1$ or $P_2$ ( obtained in Appendix \ref{app2}), irrespective of the initial point $p_{\rm in}$. Note that for both $P_1$ and $P_2$, we have
\vspace{-2mm}
\begin{align}  \label{xval}
\lim_{t \rightarrow \infty}\mathbb{E} \{ |x_{\rm sc}|(t)^2 \}&=\frac{\beta(r-1)}{\epsilon_x^2}, \nonumber \\
\text{and} \quad \lim_{t \rightarrow \infty}\mathbb{E} \{ |x_{\rm sc}(t)|^4 \}&=\frac{\beta^2(r-1)^2}{\epsilon_x^4}.
\end{align}
From Proposition \ref{prop2}, when the system is stable, i.e. $r \in \left( 1,\frac{\sigma(\sigma+\beta+3)}{\sigma-\beta-1}\right)$, irrespective of $p_{\rm in}$, the harvested DC is
\vspace{-2mm}
\begin{align}    \label{harvchaos}
\eta_{\rm SL}&= \lim_{t \rightarrow \infty} \rho_1\mathbb{E} \{ |x_{\rm sc}(t)|^2 \}+\lim_{t \rightarrow \infty} \rho_2\mathbb{E} \{ |x_{\rm sc}(t)|^4 \} \nonumber \\
&= \frac{\rho_1\beta(r-1)}{\epsilon_x^2}+\frac{\rho_2\beta^2(r-1)^2}{\epsilon_x^4},
\end{align}
which follows from \eqref{xval}.

\bibliographystyle{IEEEtran}
\bibliography{WCL1566_refs}
\end{document}